\documentclass[english,letterpaper,11pt,reqno]{amsart}

\usepackage{appendix}
\usepackage{amsbsy}
\usepackage{amsfonts}
\usepackage{amsmath}
\usepackage{amssymb}
\usepackage{amsthm}
\usepackage{graphicx}
\usepackage{ifthen}
\usepackage{textcomp}
\usepackage{centernot}
\usepackage{mathtools}
\usepackage{stmaryrd}
\usepackage{mathrsfs}
\usepackage{calrsfs}
\usepackage{enumitem, color}
\usepackage{dsfont}
\usepackage[bookmarksnumbered,colorlinks]{hyperref}
\hypersetup{colorlinks=true, linkcolor=blue, citecolor=blue}
\newtheorem{theorem}{Theorem}

\newtheorem{lemma}{Lemma}
\newtheorem{corollary}{Corollary}

\newtheorem{definition}{Definition}
\newtheorem*{example}{Example}

\newcommand{\beqa}{\begin{eqnarray}}
\newcommand{\beq}{\begin{equation}}
\newcommand{\eeqa}{\end{eqnarray}}
\newcommand{\eeq}{\end{equation}}
\newcommand{\np}{\vskip 6pt \noindent}

\newcommand\ip[3]{g({#1},{#2})}  %inner product
\newcommand{\lra}{\longrightarrow}

\newcommand{\RR}{\mathbb{R}}

\newcommand\cc{\mathscr{C}}

\newcommand\vv[1]{{\boldsymbol {#1}}} %bold vector
\newcommand\xx{\boldsymbol x}

\makeatletter
\newsavebox\myboxA
\newsavebox\myboxB
\newlength\mylenA
\newcommand*\xoverline[2][0.75]{%
    \sbox{\myboxA}{$\m@th#2$}%
    \setbox\myboxB\null% Phantom box
    \ht\myboxB=\ht\myboxA%
    \dp\myboxB=\dp\myboxA%
    \wd\myboxB=#1\wd\myboxA% Scale phantom
    \sbox\myboxB{$\m@th\overline{\copy\myboxB}$}%  Overlined phantom
    \setlength\mylenA{\the\wd\myboxA}%   calc width diff
    \addtolength\mylenA{-\the\wd\myboxB}%
    \ifdim\wd\myboxB<\wd\myboxA%
       \rlap{\hskip 0.5\mylenA\usebox\myboxB}{\usebox\myboxA}%
    \else
        \hskip -0.5\mylenA\rlap{\usebox\myboxA}{\hskip 0.5\mylenA\usebox\myboxB}%
    \fi}
\makeatother

\begin{document}
\title[]{The geometry of gravitational lensing magnification}
\author[]{Amir Babak Aazami and Marcus C. Werner}
\address{Kavli IPMU (WPI), UTIAS\hfill\break\indent
The University of Tokyo\hfill\break\indent
Kashiwa, Chiba 277-8583, Japan}
\email{amir.aazami@ipmu.jp}

\address{Yukawa Institute for theoretical physics\hfill\break\indent
Kyoto University\hfill\break\indent
Kyoto 606-8502, Japan}
\email{werner@yukawa.kyoto-u.ac.jp}

\maketitle
\begin{abstract}
We present a definition of unsigned magnification in gravitational lensing valid on arbitrary convex normal neighborhoods of time oriented Lorentzian manifolds.  This definition is a function defined at any two points along a null geodesic that lie in a convex normal neighborhood, and foregoes the usual notions of lens and source planes in gravitational lensing.  Rather, it makes essential use of the van Vleck determinant, which we present via the exponential map, and Etherington's definition of luminosity distance for arbitrary spacetimes.  We then specialize our definition to spacetimes, like Schwarzschild's, in which the lens is compact and isolated, and show that our magnification function is monotonically increasing along any geodesic contained within a convex normal neighborhood.
\end{abstract}

%%%%%%%%%%%%%%%%%%
\section{Introduction}
Flux is an important observable in gravitational lensing. Suppose $F$ is the observed flux of a light source, which has been increased due to the gravitational focusing by an intervening massive lensing object, and that the hypothetical flux in the absence of the lens would be $F_0$. Then the magnification factor due to gravitational lensing is given by
\beqa
\label{eqn:mag}
\mu\ =\ \frac{F}{F_0},
\eeqa
ignoring image orientation.  This magnification factor is not regarded as an observable because the intrinsic flux $F_0$ without the lens is unknown in general. The recent report of the first strongly lensed type Ia supernova \cite{quimby13} proves an exception to this rule because such supernovae are standardizable candles. Moreover, it is clear that many more examples of strongly lensed type Ia supernovae will be found by upcoming surveys, so cases of observable lensing magnification will become more common.
\np
From a theoretical point of view, the magnification factor $\mu$ is usually considered in terms of the quasi-Newtonian approximation for gravitational lensing (see, e.g., \cite{schneider92,petters}), which treats light rays as piecewise straight lines in Euclidean space. However, the proper arena for gravitational lensing is, of course, the Lorentzian spacetime geometry of General Relativity. Thus, it is desirable to generalize the definition of magnification to a spacetime setting and better understand its geometrical meaning.  The aim of this article is to do this by reinterpreting the classic definition of luminosity distance in spacetimes due to Etherington \cite{etherington33} in terms of lensing magnification and the van Vleck determinant \cite{van28}, which we express here in terms of the exponential map.  Indeed, we are not the first to promote the use of the exponential map in gravitational lensing; see, e.g., the viewpoint offered in \cite{reimberg13}.  For the spacetime view of gravitational lensing in general, see, e.g., \cite{perlick00}.

\np
In our understanding of the van Vleck determinant, we have been greatly aided by the two very comprehensive treatments \cite{Visser,Poisson}.  Indeed, most of Sections \ref{sec:wf} and \ref{sec:vv} below, on Synge's world function and the van Vleck determinant, can be found in \cite{Visser} and \cite{Poisson}, with one exception, however, that warrants their inclusion here: namely, our use of the exponential map to compute the van Vleck determinant.  For this reason, our notation, and several of our proofs, are noticeably different from most of the existing literature.  Our definition of magnification (Definition \ref{def:mag}) and subsequent focusing theorem (Theorem \ref{cor:ft}) are to be found in Section \ref{sec:magdef}.

%%%%%%%%%%%%%%%%%%
\section{Overview of the exponential map and normal coordinates}
\label{sec:exp}
Let $(M,g)$ be a connected time oriented four-dimensional Lorentzian manifold, with $g$ having signature $(-,+,+,+)$.  Let $\cc$ be a \emph{convex normal neighborhood of $p \in M$,} that is, a neighborhood any two points $q,q'$ of which are connected by a \emph{unique} distance-minimizing geodesic $\alpha_{qq'}$ lying entirely in $\cc$, though there may well be other geodesics between $q$ and $q'$ that leave and then reenter $\cc$; the $\alpha_{qq'}$ are usually referred to as ``radial geodesics," and we will adopt this terminology henceforth.  Furthermore, one can arrange it so that for each $q \in \cc$, $\cc$ is contained in the normal neighborhood of $q$ provided by the exponential map at $q$; that such a $\cc$ exists at every point on a Lorentzian manifold is proved, e.g., in \cite[pp.~133\,-136]{hicks}.  Thus normal coordinates $(x^i)$ centered at any point $p \in \cc$ cover all of $\cc$, and we will in fact describe Synge's ``world function" below primarily in terms of such coordinates.  Next, ``$\alpha_{qq'}$" will always denote the unique radial geodesic in $\cc$ from $q$ to $q'$, ``$\alpha_{qq'}'(t)$" will denote its tangent vector at the point $\alpha_{qq'}(t)$, and we henceforth adopt the convention of parametrizing our radial geodesics to run for unit affine parameter: $\alpha_{qq'}\colon [0,1] \lra \cc$, with $\alpha_{qq'}(0) = q$ and $\alpha_{qq'}(1) = q'$.  Also, we adopt the convention of writing both a curve, and its coordinate representation in a coordinate chart, using the same symbol.  Finally, the Einstein summation convention will be assumed throughout, with indices labeled $0,1,2,3$.

\np
Because of the essential role played for us by the exponential map, we now briefly review some of its properties.   Thus, fix $p \in \cc$ and recall that any point $q \in \cc$ is given by $q = \gamma_V(1) := \text{exp}_p(V)$ for some unique vector $V \in T_pM$, where $\gamma_V$ is the unique geodesic starting at $p$ in the direction $V$, and where $\text{exp}_p$ denotes the exponential map at $p$.  Since $\text{exp}_p$ is a diffeomorphism from a neighborhood of $0 \in T_pM$ to $\cc$, any choice of orthonormal basis $\{E_0|_p,\dots,E_3|_p\} \subset T_pM$ provides us with ``normal" coordinates $(x^i)$ defined with respect to that basis.  Indeed, expressing any $X \in T_pM$ as $X = X^iE_i|_p$, the diffeomorphism $E\colon T_pM \lra \RR^n$ sending $X \mapsto (X^0,\dots,X^3)$ composes with $\exp_p^{-1}$ to give
\beqa
\label{eqn:nc1}
q = \gamma_V(1)\ \xmapsto{\,\text{exp}_p^{-1}}\ V\ =\ V^iE_i|_p\ \xmapsto{\ E\ }\ \underbrace{(V^0,\dots,V^3)\ :=\ x(q)}_{(x^i)~\text{coordinates of}~q}.
\eeqa
Thus $x^i(q) = V^i$, and because geodesics $\gamma_V$ have the scaling property $\gamma_V(t) = \gamma_{tV}(1)$ whenever either side is defined, the geodesic $\gamma_V$ in the coordinates $(x^i)$ is given by
\beqa
\label{eqn:normal0}
\gamma_V(t)\ =\ (tV^0,\dots,tV^3).
\eeqa
Of course, $\gamma_V\colon[0,1] \lra \cc$ must be our radial geodesic $\alpha_{pq}\colon [0,1] \lra \cc$.  Furthermore, the normal coordinate basis vectors $\{\partial/\partial x^0,\dots,\partial/\partial x^3\}$ defined with respect to the $(x^i)$ satisfy, by construction,
$$
\frac{\partial}{\partial x^i}\bigg|_p\ =\ E_i|_p,
$$
hence $g_{ij}(p) = \text{diag}\,(-1,1,1,1)$ (i.e., they are ``normal" at $p$).  Because of \eqref{eqn:normal0}, it also follows that $\Gamma^i_{\,jk}(p) = 0$, hence also $\partial_i|_p(g_{jk}) = 0$; consult, e.g., \cite[Prop.~33, p.~73]{o1983}.

\np
Now we use normal coordinates at $p$ to define ``quasi-normal" coordinates at any other point $q \in \cc$, as follows.  Let $\{E_0|_p,\dots,E_3|_p\} \subset T_pM$ denote the orthonormal basis with respect to which the normal coordinates $(x^i)$ at $p$ are defined.  Given any other point $q \in \cc$, let $J_i$ denote the unique Jacobi field along the radial geodesic $\alpha_{pq}\colon [0,1] \lra \cc$ satisfying $J_i(0) = 0$ and $J_i'(0) = E_i|_p$, where ``$J_i'$" denotes the covariant derivative of $J_i$ along $\alpha_{pq}$.  Observe that $\{J_0(1),\dots,J_3(1)\} \subset T_qM$ will be linearly independent provided that $q$ is not conjugate to $p$ along $\alpha_{pq}$.  Indeed this must be so, since $\cc$ is contained in the normal coordinate chart centered at $p$, \emph{no} point in $\cc$ is conjugate to $p$ along a radial geodesic through $p$ (this is an important point to which we will return later; consult, e.g., \cite[Prop.~10, p.~271]{o1983}).  Thus $\{J_0(1),\dots,J_3(1)\} \subset T_qM$ is a basis, though it need not be orthonormal.  In any case, use it to define ``quasi-normal" coordinates $(\bar{x}^i)$ centered at $q \in \cc$, via the exponential map $\text{exp}_q$ at $q$, in the same manner as in \eqref{eqn:nc1} above.  Then, by construction,
$$
\frac{\partial}{\partial \bar{x}^i}\bigg|_q\ =\ J_i(1)
$$
for each $i = 0,\dots,3$.  In fact, each
\beqa
\label{eqn:JN}
\frac{\partial}{\partial \bar{x}^i}\bigg|_q\ =\ \frac{\partial}{\partial x^i}\bigg|_q,
\eeqa
as follows.  Set $q = \text{exp}_p(V)$ and let $\phi_{E_i|_p}|_V \in T_V(T_pM)$ denote the vector canonically identified with $V \in T_pM$ (explicitly, $\phi_{E_i|_p}|_V\colon C^{\infty}(T_pM) \lra \RR$ is defined by $f \mapsto \frac{d}{dt}(f(E_i|_p + tV))\big|_{t=0}$).  Then a standard result relating exponential maps and Jacobi fields (see, e.g., \cite[Prop.~1, p.~217]{o1983}) says that 
\beqa
\label{eqn:jacfield}
d_{V}(\text{exp}_p)(\phi_{E_i|_p}|_V)\ =\ J_i(1),
\eeqa
where $d_{V}(\text{exp}_p)\colon T_V(T_pM) \lra T_qM$ is the differential of $\text{exp}_p$ at $V$.  Finally, observe that if $q' = {\rm exp}_q(W)$ and $W = \xoverline{W}^iJ_i(1) \in T_qM$, then the radial geodesic $\alpha_{qq'}\colon [0,1] \lra \cc$ is given in the coordinates $(\bar{x}^i)$ by
$\alpha_{qq'}(t) = (t\xoverline{W}^0,\dots,t\xoverline{W}^3)$.  We will make use of these coordinates in Section \ref{sec:vv} below.

%%%%%%%%%%%%%%%%%%
\section{The world function and its gradient}
\label{sec:wf}
Though, in fact, we are already in a position to define the van Vleck determinant, in this section we first define \emph{Synge's world function} \cite{synge60} and state some of its properties.  We do this primarily to connect our work with previous presentations of the van Vleck determinant appearing in the literature, in particular \cite{Poisson}, from which this section derives.

\begin{definition}
\label{def:wf1}
With notation as above, define the {\rm world function}
\beqa
\label{def:wf}
\sigma\colon \cc \times \cc \lra \RR\hspace{.2in},\hspace{.2in}\sigma(p,q)\ :=\ \frac{1}{2}\int_0^1 \ip{\alpha_{pq}'}{\alpha_{pq}'}{g}\,dt.
\eeqa
\end{definition}
Note that when $q = p$, the radial geodesic $\alpha_{pp}$ is just the constant curve $\alpha_{pp}(t) = p$, so that $\alpha_{pp}'(t) = 0 \in T_pM$ and $\sigma(p,p) = 0$.  In fact, $\sigma$ is a smooth function  (see \cite[Lemma~9, p.~131]{o1983}) that is invariant under a change of parametrization, in the following sense: if we instead parametrize our radial geodesic $\alpha_{pq}$ to have domain $[a,b]$ and define
$$
\sigma(p,q)\ :=\ \frac{b-a}{2}\int_a^b \ip{\alpha_{pq}'}{\alpha_{pq}'}{g}\,dt,
$$
then \emph{this} function is invariant under any geodesic reparametrization of $\alpha_{pq}\colon[a,b] \lra \cc$.  Note that if $L(\alpha_{pq})$ denotes the length of $\alpha_{pq}\colon[0,1] \lra \cc$, then
$$
\sigma(p,q)\ =\
\left\{
\begin{array}{ll}
 L(\alpha_{pq})^2/2\hspace{.52in}\alpha_{pq}~\text{spacelike},\\
 0\hspace{1.103in}\alpha_{pq}~\text{null},\label{eqn:length}\\
 -L(\alpha_{pq})^2/2\hspace{.405in}\alpha_{pq}~\text{timelike}.\\
\end{array}
\right.
$$
Moving on, observe that \eqref{def:wf} is very similar to the {\it energy function} of any curve segment $\alpha\colon [0,1] \lra \cc$,
$$E_{\alpha}\colon[0,1] \lra \RR\hspace{.2in},\hspace{.2in}E_\alpha(t)\ :=\ \frac{1}{2}\int_0^t \ip{\alpha'(u)}{\alpha'(u)}{g}\,du
$$
(see \cite[pp.~288ff.]{o1983}).  Indeed, let $\vv{x}\colon[0,1] \times (-\delta,\delta) \lra \cc$ be a smooth variation of $\alpha_{pq}$ through geodesics starting at $p$; i.e., each curve $\vv{x}(\cdot,v)\colon [0,1] \lra \cc$ is a geodesic starting at $p$.  Denoting the tangent vectors of each $\vv{x}(\cdot,v)$ by $d\vv{x}(\partial/\partial t|_{(\cdot,v)}) := \vv{x}_t(\cdot,v)$, we define the energy $E_{\vv{x}}\colon (-\delta,\delta) \lra \RR$ of the variation $\vv{x}$ of $\alpha_{pq}$ to be the smooth function
$$
E_{\vv{x}}(v)\ :=\ \frac{1}{2}\int_0^1 \ip{\vv{x}_t(\cdot,v)}{\vv{x}_t(\cdot,v)}{g}\,dt.
$$
Then it is straightforward to verify that $\sigma(p,\cdot) \circ \vv{x}(1,\cdot) \colon (-\delta,\delta) \lra \RR$ satisfies
\beqa
\label{eqn:E}
\sigma(p,\cdot) \circ \vv{x}(1,\cdot)\ =\ E_{\vv{x}}.
\eeqa
If our parametrization domain were $[a,b]$, then \eqref{eqn:E} would be $\sigma(p,\cdot) \circ \vv{x}(1,\cdot) = (b-a)E_{\vv{x}}$.  This fact plays an important role in the proof of the following fundamental result regarding the world function $\sigma(p,\cdot)\colon \cc \lra \RR$ with initial point $p$ fixed. 

\begin{lemma}
\label{lemma:sigma1}
Fix $p \in \cc$ and let $\alpha_{pq}\colon[0,1]\lra \cc$ be the radial geodesic in $\cc$ from $p$ to an arbitrary point $q \in \cc$.  In normal coordinates $(x^i)$ centered at $p$, the gradient ${\rm grad}\,\sigma(p,\cdot) := K$ at $q$ is
\beqa
\label{eqn:energyvar}
K_q\ =\ \alpha_{pq}'(1)\ =\ x^i(q)\,\frac{\partial}{\partial x^i}\bigg|_q\cdot
\eeqa
Moreover, $K$ satisfies
\beqa
\label{eqn:fund1}
\ip{K}{K}{g}\ =\ 2\sigma(p,\cdot).
\eeqa
\end{lemma}

\begin{proof}
Throughout this proof, we follow the notation and terminology in \cite[pp.~215ff.]{o1983}.  If $x(q) = (x^0(q),\dots,x^3(q))$ in normal coordinates $(x^i)$ centered at $p$, then
\beqa
\label{eqn:alphanorm}
\alpha_{pq}(t)\ =\ (t\,x^0(q),\dots,t\,x^3(q)).
\eeqa
Now let  $W_q = W^i\,\partial_i|_q \in T_qM$ be arbitrary and consider the variation $\vv{x}\colon [0,1] \times (-\delta,\delta) \lra \cc$ defined in the normal coordinates $(x^i)$ by
$$
\vv{x}(t,v)\ :=\ \big(t(x^0(q) + vW^0),\dots,t(x^3(q) + vW^3)\big).
$$
This is well-defined and contained in $\cc$ for $\delta$ small enough, with all longitudinal curves $\vv{x}(\cdot,v)\colon[0,1] \lra \cc$ being geodesics fixed at $p$, base curve $\vv{x}(\cdot,0)\colon [0,1] \lra \cc$ our original geodesic $\alpha_{pq}$ in \eqref{eqn:alphanorm}, and with variation field $V\colon [0,1] \lra TM$ given by
$$
d\vv{x}\bigg(\frac{\partial}{\partial v}\bigg|_{(t,0)}\bigg)\ :=\ V(t)\ =\ (tW^0,\dots,tW^3).
$$
Because the variation is fixed at $p$,
$$
V(0)\ =\ 0 \in T_pM,
$$
while at the opposite end
$$
V(1)\ =\ W_q \in T_qM.
$$
Note that $V$ is a Jacobi field because all longitudinal curves are geodesics.  To derive \eqref{eqn:energyvar}, we now consider the differential of the composition
$$
\sigma(p,\cdot) \circ \vv{x}(1,\cdot)\colon (-\delta,\delta) \lra \RR\hspace{.2in},\hspace{.2in}v\ \mapsto\ \sigma(p,\vv{x}(1,v)),
$$
and find that at $v = 0$,
\beqa
d_0(\sigma(p,\cdot) \circ \vv{x}(1,\cdot))\bigg(\frac{d}{dv}\bigg|_{0}\,\bigg) &=& d_{\vv{x}(1,0)}(\sigma(p,\cdot)) V(1)\nonumber\\
&=& \ip{K_q}{W_q}{g}.\nonumber
\eeqa
The significance of this can be seen once we realize that, thanks to \eqref{eqn:E},
$$
(\sigma(p,\cdot) \circ \vv{x}(1,\cdot))(v)\ =\ E_{\vv{x}}(v),
$$
and the variational properties of the latter are well known.  Indeed, for the variation $\vv{x}$,
$$
\frac{dE_{\vv{x}}}{dv}\bigg|_{0}\ =\ \ip{\alpha_{pq}'}{V}{g}\,\bigg|_0^1\ =\ \ip{\alpha_{pq}'(1)}{W_q}{g},
$$
(for the full form of $E_{\vv{x}}'(0)$, see \cite[Prop.~39, p. 289]{o1983}).  Noting that our choice of $W_q \in T_qM$ was arbitrary, it follows that we must have
$$
K_q\ =\ \alpha_{pq}'(1)\ =\ x^i(q)\,\frac{\partial}{\partial x^i}\bigg|_q\cdot
$$
We point out here that if $\alpha_{pq}\colon[a,b]\lra \cc$, then the right-hand side of \eqref{eqn:energyvar} would be scaled by a factor $(b-a)$.  Working in normal coordinates centered at $p$ so that we can rely on \eqref{eqn:energyvar}, we evaluate $\ip{K}{K}{g}$ at $q \in \cc$ and obtain
$$
\ip{K}{K}{g}(q)\ =\ \ip{\alpha_{pq}'(q)}{\alpha_{pq}'(q)}{g}\ =\ 2\,\sigma(p,q),
$$
where we note that the function $t \mapsto \ip{\alpha_{pq}'(t)}{\alpha_{pq}'(t)}{}$ is constant along the geodesic $\alpha_{pq}$.
\end{proof}

The following fact, whose proof we omit, now follows as a consequence.

\begin{corollary}
\label{cor:geodesics}
Fix $p \in \cc$ and let $K = {\rm grad}\,\sigma(p,\cdot)$.  Then $K_q \neq 0$ for all $q \in \cc-\{p\}$, while $K_p = 0$.  Furthermore, $K$ satisfies
$$
\nabla_KK\ =\ K,
$$
hence its integral curves admit parametrizations as geodesics.
\end{corollary}

Before investigating the integral curves of the vector field $K$, we ask: In Lemma \ref{lemma:sigma1}, what if we had wanted to fix the endpoint $q$ instead, so that we would be working with $\sigma(\cdot,q)\colon \cc \lra \RR$?  Not surprisingly, the answer is the same, as we state here without proof.

\begin{corollary}
\label{cor:2}
Fix $q \in \cc$ and let $\alpha_{pq}\colon[0,1]\lra \cc$ be the radial geodesic in $\cc$ from an arbitrary point $p \in \cc$ to $q$.  In normal or quasi-normal coordinates $(\bar{y}^i)$ centered at $q$, the gradient ${\rm grad}\,\sigma(\cdot,q) := \xoverline{K}$ at $p$ is
\beqa
\label{eqn:energyvar2}
\xoverline{K}_p\ =\ -\alpha_{pq}'(0)\ =\ \bar{y}^i(p)\,\frac{\partial}{\partial \bar{y}^i}\bigg|_p\cdot
\eeqa
Moreover, $\xoverline{K}$ satisfies
\beqa
\label{eqn:fund2}
\ip{\xoverline{K}}{\xoverline{K}}{g}\ =\ 2\sigma(\cdot,q).
\eeqa
\end{corollary}

We close this section by describing the integral curves of $K = \text{grad}\,\sigma(p,\cdot)$.  Because $K_p = 0$ by Corollary \ref{cor:geodesics}, the integral curve of $K$ through $p$ is simply the constant curve at $p$.  Now pick any other point $q \in \cc$ and let $\gamma\colon I \lra \cc$ denote the maximal integral curve of $K$ starting at $q$, where $I \subset \RR$ is a connected open interval containing 0.  In terms of normal coordinates $(x^i)$ centered at $p$, let $\gamma(t) = (\gamma^0(t),\dots,\gamma^3(t))$ and $x(q) = (x^0(q),\dots,x^3(q))$.  Then
$$
\gamma'(t)\ =\ \dot{\gamma}^i(t) \frac{\partial}{\partial x^i}\bigg|_{\gamma(t)}\ =\ K_{\gamma(t)}\ =\ \gamma^i(t)\frac{\partial}{\partial x^i}\bigg|_{\gamma(t)},
$$
where the last equality arises via \eqref{eqn:energyvar} in Lemma \ref{lemma:sigma1}.  Since $\gamma(0) = (x^1(q),\dots,x^n(q))$, it follows that in normal coordinates $(x^i)$ centered at $p$, the integral curve of $K$ starting at $q$ is
\beqa
\label{eqn:IC}
\gamma(t)\ =\ (e^t\,x^0(q),\dots,e^t\,x^3(q)).\nonumber
\eeqa
Hence the maximal interval $I$ necessarily contains $(-\infty,0]$ and
\beqa
\label{eqn:limit}
\lim_{t \to -\infty} \gamma(t)\ =\ (0,\dots,0)\ =\ p,\nonumber
\eeqa
though $p \notin \gamma(I)$.  Of course, $\nabla_{\gamma'}\gamma' = \gamma'$ by Corollary \ref{cor:geodesics}, hence $\gamma(t)$ admits a reparametrization as a geodesic.  It is straightforward to verify that one such reparametrization is given by the diffeomorphism $h\colon I \lra h(I) \subset (0,+\infty)$ defined by $h(t) = e^t := s$, in terms of which the reparametrized curve $\tilde{\gamma}(s) := (\gamma \circ h^{-1})(s)$ is given, once again in normal coordinates $(x^i)$ centered at $p$, by
\beqa
\label{eqn:geod}
\tilde{\gamma}(s)\ =\ (s\,x^0(q),\dots,s\,x^3(q))\hspace{.2in},\hspace{.2in}\tilde{\gamma}'(s)\ =\  \frac{K_{\gamma(\text{ln}\,s)}}{s}\cdot\nonumber
\eeqa
Observe that $\tilde{\gamma}$ is certainly defined on $(0,1] = h((-\infty,0])$, on which it must coincide with the radial geodesic $\alpha_{pq}\colon(0,1] \lra \cc$, and that $\tilde{\gamma}(1) = q$.

%%%%%%%%%%%%%%%%%%
\section{The van Vleck determinant}
\label{sec:vv}
The quantity known as the \emph{van Vleck determinant} \cite{van28} arises via a simple application of the Jacobian function of a smooth mapping (see \cite[p. 196]{o1983}).  Recall that, given smooth oriented manifolds $M$ and $N$ with corresponding volume forms $\omega_M$ and $\omega_N$, and a smooth map $\phi\colon M \lra N$, the \emph{Jacobian function of $\phi$} is the smooth function $\mathcal{J}(\phi) \in C^{\infty}(M)$ such that
\beqa
\label{eqn:JN0}
\phi^*(\omega_N)\ =\ \mathcal{J}(\phi)\,\omega_M,
\eeqa
where $\phi^*$ is the pullback of $\phi$.  Now let $\phi$ be the diffeomorphism $\exp_p\colon \text{exp}_p^{-1}(\cc) \subset T_pM \lra \cc$.  The volume element on $\cc$, in terms of normal coordinates $(x^i)$ centered at $p$, is $\sqrt{-\text{det}\,[g_{ij}]}\,dx^0 \wedge \cdots \wedge dx^3$, while the volume element on $\text{exp}_p^{-1}(\cc) \subset T_pM$, which we identify via the diffeomorphism $E$ in \eqref{eqn:nc1} with an open subset in $\RR_1^4$, is simply $dx^0 \wedge \cdots \wedge dx^3$.  Then \eqref{eqn:JN0} becomes
$$
\text{exp}_p^*\left(\sqrt{-\text{det}\,[g_{ij}]}\ dx^0 \wedge \cdots \wedge dx^3\right)\ =\ \mathcal{J}(\text{exp}_p)\ dx^0 \wedge \cdots \wedge dx^3.
$$
For any $q \in \cc$ with corresponding coordinates $x(q)$, observe that $d\,\text{exp}_p(\partial/\partial x^j|_{x(q)}) = \partial/\partial x^i|_q$, hence 
\beqa
\label{eqn:vvJac}
\mathcal{J}(\text{exp}_p)(q)\ =\ \sqrt{-\text{det}\,[g_{ij}(q)]},
\eeqa
where, by identifying $q = \text{exp}_p(V)\in \cc$ with its unique vector $V \in T_pM$, we can regard $\mathcal{J}(\text{exp}_p)$ as a smooth function on $\cc$, rather than on $T_pM$.  Note that, using the quasi-normal coordinates of Section \ref{sec:exp}, and \eqref{eqn:JN} in particular, the right-hand side of \eqref{eqn:vvJac} is also given by
\beqa
\label{eqn:vvJ}
\sqrt{-\text{det}\,[\ip{J_i(1)}{J_j(1)}{}]}.
\eeqa

Indeed, recalling that $\gamma_V(t) = \gamma_{tV}(1)$ whenever either side is defined, it is straightforward to show that for all $t > 0$ at which $\gamma_V$ is defined, the differential $d_{tV}(\text{exp}_p)\colon T_{tV}(T_pM) \lra T_{\gamma_V(t)}M$ satisfies
$$
\phi_{E_i|_p}|_{tV}\ \mapsto\ d_{tV}(\text{exp}_p)(\phi_{E_i|_p}|_{tV})\ =\ \frac{J_i(t)}{t},
$$
whose derivation is similar to that in \cite[Prop.~1, p.~217]{o1983}.  Then the Jacobian function of $\text{exp}_p$ along $\gamma_V$ satisfies
$$
\mathcal{J}(\text{exp}_p)(\gamma_V(t))\ =\ \frac{\sqrt{-\text{det}\,[\ip{J_i(t)}{J_j(t)}{}]}}{t^4}\cdot
$$  

\begin{definition}[{\rm van Vleck determinant in a convex normal neighborhood}]
\label{def:vleck2}
Let $\cc$ be a convex normal neighborhood of a Lorentzian manifold $(M,g)$.  The {\rm van Vleck determinant} $\Delta\colon \cc \times \cc\lra \RR$ is the function defined by
\beqa
\label{eqn:vleck2}
\Delta(p,q)\ :=\ \frac{1}{\mathcal{J}({\rm exp}_p)(q)}\cdot
\eeqa
\end{definition}

Let us make five remarks about this definition.
(1) Recall that, because we have restricted ourselves to a convex normal neighborhood $\cc$, any two points $p,q \in \cc$ are connected by a \emph{unique} geodesic segment in $\cc$, so there is no ambiguity as to the choice of geodesic.
(2) We reemphasize that using the convex normal neighborhood $\cc$ allows us to write $\mathcal{J}(\text{exp}_p)$, which is defined on $T_pM$, as a smooth function on $\cc$ itself.
(3) Bearing \eqref{eqn:vvJac} in mind, note that, although each orthonormal basis for $T_pM$ determines a normal coordinate system centered at $p$, $\Delta(p,q)$ can be evaluated with respect to any such coordinates: orthonormal bases are related by an orthogonal change of basis matrix, hence $\Delta(p,q)$ remains unchanged.
(4) Because $p = \text{exp}_p(0)$ and $d_0\text{exp}_p\colon T_0(T_pM) \lra T_pM$ sends $\phi_X|_0 \mapsto X$, the van Vleck determinant satisfies
\beqa
\label{eqn:pp}
\Delta(p,p)\ =\ 1.
\eeqa
(5) In Minkowski spacetime $\RR_1^4$, which is itself a convex normal neighborhood of each of its points, and whose van Vleck determinant we denote henceforth by ``$\Delta_0$," it is easy to verify that
\beqa
\label{eqn:mink1}
\Delta_0(p,q)\ =\ 1\hspace{.2in}\forall p,q \in \RR_1^4.
\eeqa
In fact, \eqref{eqn:vleck2} is not the usual way in which the van Vleck determinant is presented.  Rather, it is defined via the world function $\sigma$.  In \cite{Poisson,Visser}, for example, the van Vleck determinant is defined via the mixed partial derivatives of $\sigma$:
$$
\frac{\partial^2 \sigma}{\partial \bar{x}^k \partial x^l}\bigg|_{(p,q)}\ :=\ \sigma_{\bar{k}l}(p,q).
$$
Specifically, it is defined to be the following quantity:
\beqa
\label{eqn:vleck0}
\Delta(p,q)\ :=\ -\frac{\text{det}\,[-\sigma_{\bar{k}l}(p,q)]}{\sqrt{-\text{det}\,[g_{ij}(p)]}\,\sqrt{-\text{det}\,[g_{\,\overline{i}\overline{j}}(q)]}}\cdot
\eeqa
As is easily verified, this quantity is independent of the coordinates $(x^i)$ at $p$ and $(\bar{x}^i)$ at $q$ used to compute it.  In which case, let us evaluate \eqref{eqn:vleck0} using normal coordinates $(x^i)$ centered at $p$ and corresponding quasi-normal coordinates $(\bar{x}^i)$ centered at $q$; doing so, and making use of \eqref{eqn:JN}, we will show that one obtains simply
\beqa
\label{eqn:master00}
\frac{1}{\sqrt{-\text{det}\,[g_{ij}(q)]}},
\eeqa
which is equal to $1/\mathcal{J}(\text{exp}_p)(q)$ in the same coordinates.  To derive \eqref{eqn:master00}, return to the quantity
$$
\frac{\partial^2 \sigma}{\partial \bar{x}^k \partial x^l}\bigg|_{(p,q)}
$$
and note that, by Corollary \ref{cor:2}, in quasi-normal coordinates $(\bar{x}^i)$ centered at $q$, $\xoverline{K} = \text{grad}\,\sigma(\cdot,q)$ satisfies
\beqa
\label{eqn:qfixed}
\xoverline{K}_p\ =\ \underbrace{g^{\,\overline{i}\overline{j}}(p)\,\frac{\partial \sigma(\cdot,q)}{\partial \bar{x}^i}\bigg|_p}_{\xoverline{K}_p^j}\,\frac{\partial}{\partial \bar{x}^j}\bigg|_p\ =\ \bar{x}^k(p)\,\frac{\partial}{\partial \bar{x}^k}\bigg|_p\ =\ -\alpha_{pq}'(0).
\eeqa
Now introduce normal coordinates $(x^i)$ centered at $p$, so that we can express $-\alpha_{pq}'(0)$ in two ways:
$$
\underbrace{\bar{x}^k(p)\,\frac{\partial}{\partial \bar{x}^k}\bigg|_p\,}_{q~\text{coordinates}}\ =\ -\alpha_{pq}'(0)\ =\ \underbrace{-x^k(q)\,\frac{\partial}{\partial x^k}\bigg|_p\,}_{p~\text{coordinates}}\ =\ -x^k(q)\,\frac{\partial \bar{x}^j}{\partial x^k}\bigg|_p\frac{\partial}{\partial \bar{x}^j}\bigg|_p\cdot
$$
Together with \eqref{eqn:qfixed}, this yields
\beqa
\label{eqn:almost}
\frac{\partial \sigma(\cdot,q)}{\partial \bar{x}^i}\bigg|_p\ =\ -g_{\,\overline{i}\overline{j}}(p)\,x^k(q)\,\frac{\partial \bar{x}^j}{\partial x^k}\bigg|_p\cdot
\eeqa
With $q$ fixed, $\sigma(\cdot,q)$ is a function of $p \in \cc$; since both coordinate charts $(x^i)$ and $(\bar{x})^i$ cover $\cc$, let us ``switch charts" and write the left-hand side as
$$
\frac{\partial \sigma(\cdot,q)}{\partial \bar{x}^i}\bigg|_p\ =\ \frac{\partial x^l}{\partial \bar{x}^i}\bigg|_p\frac{\partial \sigma(\cdot,q)}{\partial x^l}\bigg|_p,
$$
so that \eqref{eqn:almost} can be rewritten as
\beqa
\label{eqn:almost2}
\frac{\partial \sigma(\cdot,q)}{\partial x^l}\bigg|_p\ =\ -x^k(q)\underbrace{\,g_{\,\overline{i}\overline{j}}(p)\,\frac{\partial \bar{x}^j}{\partial x^k}\bigg|_p\frac{\partial \bar{x}^i}{\partial x^l}\bigg|_p\,}_{g_{lk}(p)}\ =\ -g_{lk}(p)\, x^k(q).
\eeqa
Now we must take the partial derivative of this with respect to $\partial/\partial \bar{x}^k$, \emph{evaluated at $q$.}  But as we saw in \eqref{eqn:JN}, quasi-normal coordinates satisfy
$$
\frac{\partial}{\partial \bar{x}^i}\bigg|_q\ =\ \frac{\partial}{\partial x^i}\bigg|_q,
$$
in which case the application of $\partial/\partial \bar{x}^k|_q$ to \eqref{eqn:almost2} yields simply
\beqa
\label{eqn:bingo1}
\frac{\partial}{\partial \bar{x}^i}\bigg|_q\left(\frac{\partial \sigma}{\partial x^l}\right)\ =\ \frac{\partial}{\partial x^i}\bigg|_q\left(\frac{\partial \sigma}{\partial x^l}\right)\ =\ -g_{lk}(p),
\eeqa
so that, in the end,
\beqa
\label{eqn:bingo}
\frac{\partial^2 \sigma}{\partial \bar{x}^k \partial x^l}\bigg|_{(p,q)}\ =\ -g_{lk}(p).
\eeqa
With this convenient formulation of the mixed partial derivatives $\sigma_{\bar{k}l}$ of $\sigma$, consider \eqref{eqn:vleck0} once again.  Evaluating it using normal coordinates $(x^i)$ centered at $p$ and corresponding quasi-normal coordinates $(\bar{x}^i)$ centered at $q$, we obtain, with the help of \eqref{eqn:bingo} and \eqref{eqn:JN}, simply
$$
\frac{1}{\sqrt{-\text{det}\,[g_{ij}(q)]}},
$$
which is precisely \eqref{eqn:master00}.  Nevertheless, the reason we have emphasized the Jacobian function of the exponential map in Definition \ref{def:vleck2} is because the exponential map gives a precise indication of when $q$ is \emph{conjugate} to $p$:

\begin{lemma}
\label{cor:conjugate}
For any $p \in M$ and $V \in T_pM$, let $\gamma_V$ denote the geodesic starting at $p$ in the direction $V$.  Then the van Vleck determinant $\Delta(p,\cdot)$ is unbounded at a point $q$ along $\gamma_V$ if and only if $q$ is conjugate to $p$ along $\gamma_V$.
\end{lemma}

\begin{proof}
Let us suppose that $q = \text{exp}_p(V) = \gamma_V(1)$.  If $q$ is conjugate to $p$ along $\gamma_V$, then by definition there exists a nontrivial Jacobi field $J$ along $\gamma_V\colon[0,1] \lra M$ satisfying $J(0) = J(1) = 0$ (see \cite[Prop.~10, p. 271]{o1983}).  Set $J'(0) = W$ (note that $W \in T_pM$ must be nonzero, since $J$ is nontrivial).   Then $d_{V}(\text{exp}_p)(\phi_{W}|_{V}) = J(1) = 0$.  Scaling $J$ so that $W$ has unit length, it follows from \eqref{eqn:vvJ} that with respect to any orthonormal basis for $T_pM$ containing $W$, $\mathcal{J}(\text{exp}_p)(q) = 0$, hence $\Delta(p,\cdot)$ is unbounded at $q$.  Conversely, suppose that $\Delta(p,\cdot)$ is unbounded at $q$; if the $J_i(1)$'s were linearly independent, hence a basis for $T_qM$, then the nondegeneracy of $g_q \colon T_qM \times T_qM \lra \RR$ would ensure that $\text{det}\,[\ip{J_i(1)}{J_j(1)}{}] \neq 0$.  Since this is not the case, it follows that the $J_i(1)$'s must be linearly dependent; hence there exists a nontrivial linear combination $\sum_{i=1}^n a_i\,J_i$ such that $\sum_{i=1}^n a_i\,J_i(1) = 0$ (each $a_i \in \RR$).  As any linear combination of Jacobi fields is again a Jacobi field, it follows that $\widetilde{J} := \sum_{i=1}^n a_i\,J_i$ is a nontrivial Jacobi field satisfying $\widetilde{J}(0) = \widetilde{J}(1) = 0$ along $\gamma_V\colon[0,1]\lra M$.  Hence $p$ and $q$ are conjugate.
\end{proof}

We add immediately that by the Morse Index Theorem for null geodesics (see \cite[Theorem~10.77, p.~398]{beem}), \emph{a null geodesic has only finitely many conjugate points.}  Also, expanding the components of the metric $g$ in normal coordinates $(x^i)$ centered at $p$, for points $q = \gamma_V(1) \in \cc$ we have
$$
-\text{det}\,[g_{ij}(q)]\ =\ 1 - \frac{1}{3}\,\text{Ric}\,(V,V)\, +\, \mathcal{O}(|V|^3),
$$
hence the lowest order expansion of the van Vleck determinant is
$$
\Delta(p,\gamma_V(1))\ =\ \frac{1}{\mathcal{J}({\rm exp}_p)(\gamma_V(1))}\ \approx\ 1 + \frac{1}{6}\,\text{Ric}(V,V),
$$
cf. \cite[Eqn. (64)]{Visser} and \cite[Eqn. (7.3)]{Poisson}.

%%%%%%%%%%%%%%%%%%
\section{Magnification in gravitational lensing}
\label{sec:magdef}
Now we come to gravitational lensing in spacetimes.  What we need to define is a magnification function of two points along a null geodesic in a spacetime, the latter modeling a light ray in whose magnification we are interested.  As this null geodesic is emitted by a light-emitting source, we must therefore view our initial point $p$ as lying on the worldline of that source, that is, on a smooth, future-pointing timelike curve $\gamma(\tau)$, which, for convenience, we take to start at $\gamma(0) = p$ in $\cc$.  Let $q \in \cc$ and consider the radial geodesic $\alpha_{pq}\colon[0,1] \lra \cc$.  First, we show that the quantity 
$$
\frac{d\sigma(\gamma(\tau),q)}{d\tau}\bigg|_{\tau=0} 
$$
is nonzero.  We will show this via a slight alteration of the proofs of Lemma \ref{lemma:sigma1} and Corollary \ref{cor:2}, as follows.  In normal coordinates $(y^i)$ centered at $q$, let $\bar{\gamma}^i(\tau)$ denote the components of $\gamma(\tau)$, so that, in particular, $y^i(p) = \bar{\gamma}^i(0)$, and define a variation $\tilde{\xx}\colon [0,1] \times (-\delta,\delta) \lra \cc$ of $\alpha_{pq}$ by
\beqa
\label{eqn:varcoord}
\tilde{\xx}(t,\tau)\ :=\ \big((1-t)\bar{\gamma}^0(\tau),\dots,(1-t)\bar{\gamma}^3(\tau)\big),
\eeqa
which is well-defined for $\delta$ small enough.  Observe that $\tilde{\xx}$ is a variation of $\tilde{\xx}(t,0) = \alpha_{pq}(t)$ through geodesics, all of which are fixed at the endpoint $q$. The crucial new input here is that the initial transverse curve now coincides with the worldline of our light-emitting source:
$$
\tilde{\xx}(0,\tau)\ =\ (\bar{\gamma}^0(\tau),\dots,\bar{\gamma}^3(\tau))\ =\ \gamma(\tau).
$$
Now we proceed as we did in Lemma \ref{lemma:sigma1} and Corollary \ref{cor:2} and take the derivative of the composition 
$$
\sigma(\cdot,q) \circ \tilde{\xx}(0,\cdot)\colon (-\delta,\delta) \lra \RR\hspace{.2in},\hspace{.2in}\tau\ \mapsto\ \sigma(\tilde{\xx}(0,\tau),q)\ =\ \sigma(\gamma(\tau),q), 
$$
at $\tau = 0$.  Doing so is made easy once again because
\beqa
\label{eqn:p1}
(\sigma(\cdot,q) \circ \gamma)(\tau)\ =\ E_{\tilde{\xx}}(\tau),
\eeqa
with the energy of $\tilde{\xx}$ satisfying
$$
\frac{dE_{\tilde{\xx}}}{d\tau}\bigg|_{0}\ =\ \ip{\alpha_{pq}'}{\widetilde{V}}{g}\,\bigg|_0^1\ =\ -\ip{\alpha_{pq}'(0)}{\bar{\gamma}'(0)}{g}\ >\ 0,
$$
where $\widetilde{V}(t)$ denotes the variation field $\tilde{\xx}_{\tau}(t,0)$, which is a Jacobi field satisfying $\widetilde{V}(0) = \gamma'(0)$ and $\widetilde{V}(1) = 0$, and where $\ip{\alpha_{pq}'(0)}{\gamma'(0)}{g} < 0$ because $\alpha_{pq}(t)$ is a future-pointing null geodesic and $\gamma(\tau)$ is a future-pointing timelike curve.  In other words, the function $\tau \mapsto \sigma(\gamma(\tau),q)$ necessarily satisfies
\beqa
\label{eqn:pos}
\frac{d\sigma(\gamma(\tau),q)}{d\tau}\bigg|_{\tau=0}\ =\ -\ip{\alpha_{pq}'(0)}{\bar{\gamma}'(0)}{g}\ >\ 0, 
\eeqa
which is nonzero, as desired.  Observe that \eqref{eqn:pos} is, in fact, parametrization invariant; i.e., if we had worked with a different parametrization of $\alpha_{pq}$, say one of the form $\tilde{\alpha}_{pq}\colon [a,b] \lra \cc$ with $\tilde{\alpha}_{pq}(a) = p$ and $\tilde{\alpha}_{pq}(b) = q$, so that $\tilde{\alpha}_{pq}'(a) = (b-a)^{-1}\alpha_{pq}'(0)$, then in fact \eqref{eqn:p1} would be replaced by $(\sigma(\cdot,q) \circ \gamma)(\tau) = (b-a)E_{\tilde{\xx}}(\tau)$, so that the right-hand side of \eqref{eqn:pos} would remain unchanged.  With this established, we are now ready to give a spacetime definition of magnification in gravitational lensing.  To that end, recall that in Euclidean space, an isotropically emitting light source of luminosity $L$ has an observed flux $F$ at radial distance $r$ given by
$$
F\ =\ \frac{L}{4\pi r^2}\cdot
$$
Now, while both $L$ and $F$ are local properties that are meaningful in an arbitrary spacetime, the distance $r$ is not. However, one can use $L$ and $F$ to construct a distance measure that \emph{is} meaningful in spacetime, namely, the luminosity distance $D$:
\beqa
D(p,q)\ =\ \sqrt{\frac{L(p)}{4\pi F(q)}}\cdot
\label{luminosity}
\eeqa
Thus, a naive generalization of the lensing magnification to spacetime follows from (\ref{eqn:mag}) and (\ref{luminosity}),
\beqa
\mu(p,q)\ =\ \frac{D^2_0(p,q)}{D^2(p,q)},
\label{magnification2}
\eeqa
where the subscript $0$ denotes ``the absence of the lens"\,---\,of which we will have more to say below.  Now, Etherington \cite{etherington33} has shown that the luminosity distance of a light source in an arbitrary spacetime can be written as
\beqa
D(p,q)\ =\ -\frac{d\sigma(\gamma(\tau),q)}{d \tau}\bigg|_{0}\left(\frac{\text{det}\,[g_{ij}(q)]}{\text{det}\,[g_{ij}(p)]}\right)^{\frac{1}{4}}
\label{etherington}
\eeqa
\emph{provided that this is evaluated in normal coordinates $(x^i)$ centered at $p$,} whose trajectory in spacetime is given, as above, by a future pointing timelike curve $\gamma(\tau)$ satisfying $\gamma(0) = p$.  Indeed, in these coordinates $|\text{det}\,[g_{ij}(p)]| = 1$.  Therefore, if we combine \eqref{eqn:vvJac}, \eqref{eqn:vleck2}, \eqref{magnification2}, and \eqref{etherington}, we arrive at last at a definition for the lensing magnification in an arbitrary spacetime in terms of the van Vleck determinant:

\begin{definition}[Unsigned magnification]
\label{def:mag}
Let $(M,g)$ be a spacetime modeling a gravitational lens.  Let $p,q \in M$ lie in a convex normal neighborhood of $M$.  Then the {\rm unsigned magnification of $p$ at $q$} is defined to be
\beqa
\mu(p,q)\ :=\ \left(\frac{\frac{d\sigma_0}{d\tau}}{\frac{d\sigma}{d\tau}}\Bigg|_{0}\right)^{\!\!2}\,\frac{\Delta(p,q)}{\Delta_0(p,q)},
\label{result}
\eeqa
where the symbol ``0" denotes the spacetime $M$ in the absence of the lens.
\end{definition}

Observe that, because our lensing scenario takes place in a convex normal neighborhood of our spacetime, then $\Delta(p,q)$ and $(d\sigma/d\tau|_0)^{-1}$ are both finite and positive: finite by \eqref{eqn:pos}, and positive by Corollary \ref{cor:conjugate}, the latter because $p$ will have no conjugate points along any radial geodesic through it.  Because $\mu > 0$, our definition is that of \emph{unsigned} magnification; regarding the parity of lensed images, see, e.g., \cite[p.~34]{schneider92}.  Having said that, note that while $\Delta(p,q)$ and $d\sigma/d\tau|_0$ are defined on our spacetime $M$, $d\sigma_0/d\tau|_0$ and $\Delta_0(p,q)$ need not be: ``the absence of the lens" may well imply a \emph{different smooth manifold $M_0 \neq M$,} not just a different metric on $M$.  Let us give an important example of when this is the case.

\np
If the mass distribution acting as lens is modeled by, say, the Schwarzschild spacetime, then the absence of this lens simply means the limit $m \to 0$, where $m$ is the mass parameter.  But the limit $m \to 0$ is Minkowski spacetime $\RR_1^4$, for which \eqref{eqn:mink1} dictates that $\Delta_0 \equiv 1$.\footnote{The notion of a ``limit of a spacetime" is generally coordinate-\emph{dependent;} see \cite{geroch69}.}  In fact if we assume that our source and observer are ``spatially constant," then we can easily determine the quantity $d\sigma_0/d\tau|_0$ as well.  Indeed, in terms of global normal coordinates $(x^i)$ in $\RR_1^4$, let us imagine the worldline $\gamma(\tau)$ of our light-emitting source as given simply by
$$
\gamma(\tau)\ =\ (\tau,0,0,0).
$$
In other words, our light source is ``spatially constant" in the sense that its trajectory is the integral curve of the time orientation $\partial/\partial x^0$ starting at $p = \gamma(0) = \vv{0}$.  Similarly, we imagine our trajectory, as observers, as also being spatially constant and hence given by some
$$
\gamma_o(s)\ =\ (s,r_1,r_2,r_3).
$$
This scenario in $\RR_1^4$ is now as follows: at $\tau = 0$, our light-emitting source emits null geodesics in all null directions, precisely one of which will reach us, say at $\gamma_o(s_*) = q$.  This scenario then repeats itself at every $\tau > 0$ thereafter, as both source and observer go forward in $x^0$ along their respective integral curves of $\partial/\partial x^0$.  By \eqref{eqn:pos},
\beqa
\label{eqn:ed}
\frac{d\sigma_0(\gamma(\tau),q)}{d\tau}\bigg|_{0}\ =\ -\ip{\alpha_{\vv{0}q}'(0)}{\gamma'(0)}{}.
\eeqa
But since $\alpha_{\vv{0}q}\colon[0,1] \lra \RR_1^4$ is given by
$$
\alpha_{\vv{0}q}(t)\ =\ (t\,s_*,t\,r_1,t\,r_2,t\,r_3),
$$
and since it is a null geodesic, it is easily verified that $d\sigma_0/d\tau|_0$ is just the ``spatial separation" between source and observer:
\beqa
\label{eqn:ss}
\frac{d\sigma_0(\gamma(\tau),q)}{d\tau}\bigg|_{0}\ =\ \sqrt{r_1^2+r_2^2+r_3^3}\ :=\ r_s.
\eeqa
Therefore, if we consider the ``spatially-scaled" magnification $\mu(p,q)/r_s^2 := \mu_{s}(p,q)$, then this would now be a function defined solely on Schwarzschild spacetime, and which, bearing \eqref{eqn:vleck2} and \eqref{eqn:pos} in mind, inspires the following definition:  

\begin{example}[{\rm Unsigned magnification; compact, isolated body, with Minkowski limit}]
Let $(M,g)$ be a spacetime modeling a compact, isolated gravitating body of mass $m$ whose limit $m \to 0$ is Minkowski spacetime.  Let $\gamma(\tau)$ be a future-pointing timelike curve starting at $\gamma(0) = p \in M$ and $\gamma_V(t)$ any future-pointing null geodesic starting at $p$ in the direction $V$, such that $\gamma_V(1) = q$ lies in a convex normal neighborhood of $p$.  Then the {\rm spatially-scaled magnification $\mu_{s}$ of $p$ at $q$} is
\beqa
\label{def:magcompact}
\mu_{s}(p,q)\ :=\ \frac{1}{g(\gamma'(0),{\rm exp}_p^{-1}(q))^{2}}\frac{1}{\mathcal{J}({\rm exp}_p)(q)},
\eeqa
where $\mathcal{J}({\rm exp}_p)$ is the Jacobian function of the exponential map ${\rm exp}_p$.
\end{example}

Let us make three remarks regarding this example.
(1) Implicit in this definition is the assumption that the family of metrics $g(m)$ depends smoothly on $m$ and has a well defined limit as $m \to 0$; as mentioned above, such an assumption is coordinate-dependent in general.
(2) Recall our observation of the parameter invariance of \eqref{eqn:pos}, so that one cannot scale the right-hand side of $\mu_s$ by simply rescaling the initial null direction $V$.
(3) In order for \eqref{result} to be well defined, a Lorentzian manifold $M_0$ representing the situation of ``absence of lens," as well as points $p, q \in M_0$ with which the points $p,q \in M$ can be identified in that absence, must be clearly defined\,---\,though we do not pursue this here, we point out that one way of making precise the identifications of such points is via the method in \cite{geroch69}.  Now, as  \eqref{def:magcompact} shows, when $M_0 = \RR_1^4$, the determination of the points $p,q \in \RR_1^4$ is unnecessary, since $\Delta_0 \equiv 1$, and we can scale away the ``spatial separation" arising from the factor $d\sigma_0/d\tau|_0$.  Indeed, because we are assuming that the portion of the light ray in which we are interested resides entirely within a (very large) convex normal neighborhood, so that there are no conjugate points along it, then a further benefit is the following.  Before stating this result, let us formalize what we mean by a \emph{light ray} in a spacetime.

\begin{definition}[{\rm Light ray}]
\label{def:lightray}
Let $(M,g)$ be a time oriented Lorentzian manifold.  A {\rm light ray} in $M$ is a complete, future-pointing null geodesic $\tilde{\gamma}$ satisfying ${\rm Ric}(\tilde{\gamma}',\tilde{\gamma}') \geq 0$ and whose orthogonal complement $\tilde{\gamma}^{\perp}$ is everywhere integrable. 
\end{definition}

\begin{theorem}
\label{cor:ft}
In a convex normal neighborhood $\cc$ of $p$, $\mu_s(p,\cdot)$ is monotonically increasing along any light ray through $p$.
\end{theorem}

\begin{proof}
This follows from the fact that the van Vleck determinant is monotonically increasing along any light ray in $\cc$; for a proof, see \cite{Visser}.  
\end{proof}

We close this paper by pointing out two interesting directions for future work that arise within the framework that we have established here.  The first issue is to better understand the phenomenon of multiple lensed images in this setting.  In particular, it would be interesting to see if the occurrence of multiple lensed images can occur within a convex normal neighborhood.  A second related issue is to formulate a spacetime analogue of a magnification invariant associated to multiple lensed images.  Such invariant sums of the signed image magnification are known to occur for certain classes of lens systems in the standard approximation but there are, as yet, no generalizations applicable to spacetimes (see, e.g., \cite{PW} and the references therein).

\section*{Acknowledgements}
This work was supported by the Hakubi Center for Advanced Research, Kyoto University, Kyoto, Japan, and the World Premier International Research Center Initiative (WPI), MEXT, Japan.  The authors thank Eric Poisson and Matt Visser for very helpful discussions.

\bibliographystyle{siam}
\bibliography{Aazami-Werner-van-Vleck-FINAL}
\end{document}